\documentclass[sigconf,nonacm]{acmart}
\AtBeginDocument{%
  }

\setcopyright{none}
\acmISBN{978-1-4503-XXXX-X/2018/06}

\usepackage{tikz}
\usepackage[linesnumbered,ruled,vlined]{algorithm2e}
\usetikzlibrary{arrows.meta,positioning,calc,decorations.pathreplacing}
\usepackage{subcaption}
\usepackage{xcolor}
\usepackage{framed}
\usepackage{enumitem}
\usepackage{booktabs}
\usepackage{tabularx}
\usepackage{colortbl}
\usepackage{array}
\definecolor{my_color2}{rgb}{.50,.55,.10}

\definecolor{prat_colour}{rgb}{.70,.25,.70}

\newcommand{\slotfn}{\mathrm{slot}}
  \newcommand{\cyclefn}{\mathrm{cycle}}
  \SetKw{KwAnd}{and}   %

\begin{document}
\raggedbottom
\title{Packed-Ring Queues: Lock-Free and Wait-Free \\ Concurrent Queues on GPU}

\author{Pratheek Prakash Shetty}
\email{pratheekps@vt.edu}
\affiliation{%
  \institution{Department of ECE\\ Virginia Tech}
  \city{Blacksburg}
  \state{Virginia}
  \country{USA}
}

\author{Thomas R. W. Scogland}
\email{scogland1@llnl.gov}
\affiliation{
\institution{Lawrence Livermore National Laboratory}
\city{Livermore}
\state{California}
\country{USA}
}

\author{Wu-chun Feng}
\email{wfeng@vt.edu}
\affiliation{%
  \institution{Department of ECE and CS \\ Virginia Tech}
  \city{Blacksburg}
  \state{Virginia}
  \country{USA}}

\begin{abstract}

Concurrent queues are a foundational building block of parallel programming. On CPUs, decades of research have produced queue designs with well-understood correctness and progress guarantees --- from lock-free rings to 
wait-free constructions. On GPUs, where queues increasingly sit on the critical path of device-side task graphs and GPU-initiated I/O, no comparable foundation exists: SIMT execution serializes divergent retries across a warp; tens of thousands of threads contend on the same counter words; the double-width atomics that several CPU designs assume are absent; and GPU memory models are weakly ordered, so the classical algorithms do \emph{not} transfer unchanged.

We present the \emph{packed-ring queues} (PRQs) --- two linearizable, bounded GPU queues that represent each compound atomic state in one native word of at most 64 bits: (1)~\textbf{PRQ-LF}, a lock-free ring whose warp-batched fast path issues one atomic per warp on each of the queue's contended words, and (2)~\textbf{PRQ-WF}, a wait-free ring in which peer threads cooperatively complete any operation that falls behind. To the best of our knowledge, PRQ-WF is the first GPU concurrent queue 
whose wait-freedom is proved as a theorem, with no assumption about how the GPU schedules threads.
Across four GPUs from two vendors (NVIDIA A100 and GH200, AMD MI210 and MI300A), PRQ-LF outperforms the blocking Broker Queue by 5.0--20.4$\times$ at 32K threads, and PRQ-WF sustains throughput under producer/consumer imbalance where a prior design completes fewer than 0.1\% of its attempts.

\end{abstract}



\maketitle

\section{Introduction}
\label{sec:intro}

Concurrent queues are a cornerstone of shared-memory parallel processing on
CPUs, with decades of research establishing a rich design space and
well-understood notions of correctness and progress~\cite{michael1996,tsigas2001,kogan-petrank-method,kogan-petrank-queue,yang-wfq,lcrq,scq,wcq,aggfunnels}. This maturity spans lock-free compare-and-swap (CAS) based designs, scalable fetch-and-add (FAA) based rings, and practical wait-free constructions. (We define CAS, FAA, and the notions of progress guarantees precisely in \S\ref{sec:design}.)

In contrast, 
GPU concurrent queues 
remain far less mature, even
though today's exascale systems are accelerator-driven across vendors: Frontier and El Capitan on AMD MI250X GPUs and MI300A APUs, respectively~\cite{frontier,el-capitan} and JUPITER on NVIDIA GH200 superchips~\cite{jupiter}. Device-side coordination mechanisms are thus first-class primitives in today's heterogeneous machines, yet the GPU concurrent queue literature contains only a handful of designs~\cite{scogland2015,kerbl2018,troendle19,bacq}, with little attention to
formal progress guarantees.

GPU programs need concurrent queues wherever work is discovered at run time rather than fixed at launch, and recent systems increasingly place that queue in device memory on the critical path: Mustard executes multi-GPU task graphs through a device-side queue~\cite{mustard}; AGIO coordinates GPU-initiated I/O through concurrent request queues~\cite{agio}; and persistent-kernel schedulers and multi-GPU dataflow systems feed resident thread groups from shared, device-resident worklists~\cite{atos-icpp,groute}. Simply porting CPU queue designs to the GPU does \emph{not} work. Under single-instruction, multiple-thread (SIMT) execution, divergent retry loops serialize across the entire warp.\footnote{We use NVIDIA terminology in this paper. A \emph{warp} is 32 threads wide on the NVIDIA GPUs we evaluate. AMD's equivalent unit is the \emph{64}-thread wavefront on CDNA.} Tens of thousands of threads contend on the same head and tail counter words. The double-width compare-and-swap (CAS2) that several CPU queues assume is \emph{not} available on most current GPUs. Device-side allocation is impractical, making bounded memory a requirement. GPU memory models are weakly ordered~\cite{alglave}, so a port must place the fences that a CPU design gets from its host's stronger ordering.

Two requirements dominate concurrent queues design: (1)~\emph{semantics}, i.e., FIFO ordering and linearizability~\cite{linearizability}, where every operation appears to take effect atomically at some point between its invocation and response so that programs can reason about the order in which work is produced and consumed, and (2)~\emph{progress}, i.e., guarantees about completion under contention: lock-freedom (some operation always completes) and wait-freedom (every operation completes in a bounded number of its own steps). CPU queues offer both in mature forms~\cite{michael1996,tsigas2001,lcrq,scq,wcq,yang-wfq}. Prior GPU queues instead emphasize throughput and work distribution~\cite{scogland2015,kerbl2018,bacq}, and recent designs by Troendle, Polak, and Jang claim ``wait-free, retry-free'' reservation~\cite{troendle19,agile,bacq}. This property covers only the position-claiming FAA while completion still blocks on a peer (\S\ref{sec:related}). As a result, the field still lacks a \emph{GPU-aware wait-free concurrent queue with explicit, theorem-grounded progress guarantees} \emph{together with validated linearizability}.

To address this gap, we propose and design 
GPU-aware concurrent queues with strong progress guarantees under the umbrella of \emph{packed-ring queues} (PRQ), named for packing all shared slot state into single 64-bit words. Our contributions are as follows.

\begin{itemize}[leftmargin=*]
  \item \textbf{PRQ-LF, a lock-free packed ring.} PRQ-LF keeps the
  bounded-ring discipline of the scalable circular queue (sCQ)~\cite{scq}
  and rebuilds how the queue touches memory for SIMT execution. Its
  warp-batched fast path issues one atomic per warp on each of the queue's
  contended words, and it carries no slow-path overhead.

  \item \textbf{PRQ-WF, a wait-free packed ring.} PRQ-WF extends the same
  bounded ring with a cooperative slow path in the fast-path/slow-path
  tradition of wCQ~\cite{wcq} but is re-architected so that helping stays cheap
  among $10^4$ to $10^5$ participants and needs only 64-bit atomics.
  Carrying the producer identity in the 
  publishing CAS replaces wCQ's
  $O(k)$ producer search with $O(1)$ work. To the best of our knowledge, PRQ-WF is the
  first GPU queue whose wait-freedom is stated as a theorem
  (Theorem~\ref{thm:wf-wf}), and the theorem needs no assumption about
  how the GPU schedules threads.

  \item \textbf{Proofs and validated correctness.} We prove linearizability
  and progress properties for both queues and validate FIFO behavior using
  Porcupine-based linearizability checking and device-side FIFO tests.

  \item \textbf{Cross-vendor evaluation and analysis.} We evaluate our PRQs 
  versus
  the strongest available baselines: the Scogland--Feng
  Queue (SFQ)~\cite{scogland2015} and the Broker Queue of Kerbl et al.\ (BQ)
  along with its work-distributor variant BWD~\cite{kerbl2018}, using
  fixed-duration micro-benchmarks, hardware-level profiling metrics, and
  queue-backed level-synchronous BFS on NVIDIA A100 and GH200 and AMD MI210
  and MI300A.
\end{itemize}

We evaluate each queue under two structurally different micro-benchmarks and one application workload. The \emph{closed-loop kernel} alternates one enqueue with one dequeue per thread — the coupled consume-and-replace pattern of workloads such as frontier expansion. The \emph{open-loop kernel} assigns each thread a fixed producer or consumer role and sweeps the producer share to drive the queue into \emph{draining} (near-empty), \emph{rate-matched}, and \emph{filling} (near-full) conditions. Lastly, queue-backed level-synchronous breadth-first search (BFS) exercises the same pattern with graph-generated work; holding the traversal fixed and varying only the queue isolates the queue's contribution to end-to-end time.

\section{Related Work}
\label{sec:related}

\subsection{CPU Concurrent Queues}

Concurrent queue design on CPUs spans classic algorithms like CAS-based queues~\cite{michael1996,tsigas2001}, FAA-based ring queues~\cite{lcrq}, and wait-free fast-path/slow-path constructions~\cite{yang-wfq}. Michael and
Scott's queue (MSQ) is 
the seminal work in the field and
remains the benchmark for correctness and portability~\cite{michael1996}. The LCRQ explored an FAA-based queue to reduce contention on shared head and
tail updates, but it relies on linked ring segments and CAS2
synchronization~\cite{lcrq}; later work showed that an LCRQ-style queue can
be implemented without CAS2~\cite{lcrq-no-cas2}. Nikolaev introduced the
scalable circular queue (sCQ), a bounded lock-free ring design that
preserves ordering semantics while avoiding some of LCRQ's portability and
memory-management limitations~\cite{scq}.

Achieving wait-free guarantees on CPU required a different approach. The first fast-path/slow-path methodology developed to convert lock-free structures into wait-free ones through helping was introduced by Kogan and Petrank~\cite{kogan-petrank-method}, who developed a practical wait-free queue based on the Michael--Scott structure~\cite{kogan-petrank-queue}, while Yang and Mellor-Crummey (YMC) later showed that a wait-free queue can be as fast as FAA-based lock-free alternatives by combining an FAA-based fast path with helping through per-thread request records~\cite{yang-wfq}. More recently, Nikolaev and Ravindran revisited the problem from a bounded-memory perspective, arguing that practical wait-freedom should not rely on unbounded growth or deferred reclamation --- a flaw they identified in YMC's design ---and introduced wCQ,
built atop the sCQ structure~\cite{wcq}. Aggregating funnels reduce FAA hot-spot pressure~\cite{aggfunnels}; our warp-batched fast path is the SIMT-native counterpart, applied in a bounded GPU ring with explicit correctness arguments. Table~\ref{tab:queue-guarantees} summarizes the guarantees each queue claims.

\newcommand{\full}{\ensuremath{\bullet}}
\newcommand{\partialmark}{\ensuremath{\circ}}
\newcommand{\none}{\ensuremath{-}}

\definecolor{cpuqueuebg}{RGB}{235,244,252}
\definecolor{gpuqueuebg}{RGB}{235,248,238}

\begin{table}[t]
  \centering
  \scriptsize
  \setlength{\tabcolsep}{1.25pt}
  \setlength{\arrayrulewidth}{0.4pt}
  \renewcommand{\arraystretch}{1.08}

  \caption{CPU- and GPU-origin queue guarantees.}
  \label{tab:queue-guarantees}

  \begin{tabular}{@{}
    >{\raggedright\arraybackslash}p{1.85cm}
    >{\columncolor{cpuqueuebg}}c
    >{\columncolor{cpuqueuebg}}c
    >{\columncolor{cpuqueuebg}}c
    >{\columncolor{cpuqueuebg}}c
    >{\columncolor{cpuqueuebg}}c
    >{\columncolor{cpuqueuebg}}c|
    >{\columncolor{gpuqueuebg}}c
    >{\columncolor{gpuqueuebg}}c
    >{\columncolor{gpuqueuebg}}c
    >{\columncolor{gpuqueuebg}}c|
    >{\columncolor{gpuqueuebg}}c
    >{\columncolor{gpuqueuebg}}c@{}}

    \toprule
    \textbf{Property / guarantee}
    & {\fontsize{5.5}{6}\selectfont MSQ}
    & {\fontsize{5.5}{6}\selectfont TZQ}
    & {\fontsize{5.5}{6}\selectfont LCRQ}
    & {\fontsize{5.5}{6}\selectfont sCQ}
    & {\fontsize{5.5}{6}\selectfont wCQ}
    & {\fontsize{5.5}{6}\selectfont YMC}
    & {\fontsize{5.5}{6}\selectfont SFQ}
    & {\fontsize{5.5}{6}\selectfont BQ}
    & {\fontsize{5.5}{6}\selectfont BWD}
    & {\fontsize{5.5}{6}\selectfont BACQ}
    & {\fontsize{5.5}{6}\selectfont PRQ-LF}
    & {\fontsize{5.5}{6}\selectfont PRQ-WF} \\
    \midrule

    Lock-free progress
    & \full & \full & \full & \full & \full & \partialmark
    & \none & \none & \none & \none
    & \full & \full \\

    Wait-free progress
    & \none & \none & \none & \none & \full & \partialmark
    & \none & \none & \none & \none
    & \none & \full \\

    Linearizable FIFO semantics
    & \full & \full & \full & \full & \full & \full
    & \full & \full & \none & \full
    & \full & \full \\

    Bounded memory
    & \none & \full & \none & \full & \full & \partialmark
    & \full & \full & \full & \full
    & \full & \full \\

    Formal GPU progress argument
    & \none & \none & \none & \none & \none & \none
    & \none & \none & \none & \none
    & \full & \full \\

    \bottomrule
  \end{tabular}

  \vspace{1pt}

  \parbox{\columnwidth}{%
    \footnotesize
    \full\ fulfills the property;
    \partialmark\ partially or conditionally fulfills it;
    \none\ no claimed guarantee.\\[-1pt]
    \setlength{\fboxsep}{1.1pt}%
    \colorbox{cpuqueuebg}{\strut CPU-origin queue}\quad
    \colorbox{gpuqueuebg}{\strut GPU-native queue}.%
  }

\end{table}

\subsection{GPU Concurrent Queues} 

The GPU queue literature is much sparser than for CPU concurrent queues. 
The Scogland--Feng Queue (SFQ) is one of the earliest and most widely cited
GPU concurrent queues: a bounded, linearizable queue built around ticketing
in a fixed-size ring, emphasizing throughput and providing a high-throughput blocking interface together with a separate non-waiting interface~\cite{scogland2015}. An early warp-centric lock-free GPU queue by Zhang et al.\ moved leader-aggregated enqueue and dequeue into the SIMT execution model~\cite{Zhang14}. The Broker Queue (BQ) uses a shared admission counter, fetch-and-add reservations, and per-slot tickets to coordinate queue access. The BQ linearizable variant is blocking by design: the authors argue that under GPU contention, a potentially blocking algorithm is preferable to contended CAS. The authors also present a faster non-linearizable variant and a work-distributor configuration (BWD)~\cite{kerbl2018}. Since our study targets explicit progress guarantees with linearizability, we include both the linearizable BQ and BWD as baselines.

Troendle et al. 
develop \emph{retry-free} GPU queues: a persistent-thread scheduler queue in which a per-group proxy thread performs all atomics and the position-claiming FAA never fails~\cite{troendle19}--described as ``wait-free, retry-free'' in the accompanying dissertation~\cite{troendle-dis}--followed by AgileQueue~\cite{agile} and, most recently, BACQ, which adds boundary checks before the FAA and warp-leader consolidation with per-request
tickets~\cite{bacq}. We share two design points with this line: (1) FAA-based
reservation and (2) warp-level atomic aggregation, but the progress
guarantees differ: retry-freedom bounds only the reservation step, while
completion waits on the specific peer assigned an earlier queue position to publish, with no mechanism for other threads to complete a delayed peer's operation--blocking under the common notions of progress~\cite{herlihy-book}, and distinct from wait-freedom, which bounds completion in an operation's own steps. In addition, these designs do \emph{not} report linearizability checking over concurrent histories.\footnote{As none of these queues has a public implementation, we compare against Broker Queue, the strongest published baseline in their own evaluations.}

\subsection{Queues, Worklists, and Progress on GPUs}

Device-resident queues appear as load-bearing components of larger GPU systems. Mustard executes multi-GPU task graphs entirely on-device, retrieving and inserting task IDs through a device-side concurrent queue~\cite{mustard}; AGIO coordinates GPU-initiated I/O through concurrent request queues~\cite{agio}; and Groute and Atos move dynamically discovered work through device-resident worklists and queues~\cite{groute,atos-icpp}. In each case, the system inherits the semantics and progress behavior its queue provides, typically without stated guarantees. Our aim is to harden the component these systems assume. GPU scheduling constrains how such a guarantee can be built: thread blocks beyond the resident occupancy limit are not scheduled until earlier blocks retire~\cite{sorensen-oopsla16}, and vendors document no fairness bound across blocks~\cite{volta,sorensen-concur18}. Any step that waits for another thread inherits these gaps, so PRQ-WF contains no such step, and its wait-freedom (Theorem~\ref{thm:wf-wf}) holds under every schedule. GPU memory models also admit weak behaviors~\cite{alglave}, so we state the fences the proofs rely on (\S\ref{sec:setup}).

\section{Packed-Ring Queues: Lock-Free \& Wait-Free}
\label{sec:design}
\label{sec:background}
 
Before presenting our packed-ring queue (PRQ) designs -- lock-free (\S\ref{sec:prqlf}) and wait-free (\S\ref{sec:prqwf}) -- we articulate terminology that we use in our PRQ.
GPU threads execute in fixed-width groups (i.e., warps) that share an instruction stream; divergent branches within a warp execute serially and only blocks that are resident on the hardware are guaranteed to be scheduled. Atomics such as \textsf{fetch-and-add},
$\textsc{FAA}(a,v)$, atomically adds $v$ to the word at $a$ and returns the prior value, while \textsf{compare-and-swap}, $\textsc{CAS}(a,e,d)$, atomically replaces $a$ with $d$ if $a$ equals $e$. CAS2, a double-width CAS that several CPU queues assume~\cite{lcrq,wcq}, is not available on most current GPUs, so both our 
PRQ designs --- lock-free and wait-free ---
pack all concurrently modified state into single native words of at most 64-bits. 

A \emph{ring queue} of logical capacity $n$ stores entries in $2n$ physical slots; monotonic \texttt{Head} and \texttt{Tail} counters are advanced by FAA, and the returned value is the operation's \emph{ticket} $t$, mapping to slot $\slotfn(t) = t \bmod 2n$ in reuse epoch $\cyclefn(t) = \lfloor t/2n \rfloor$, the slot's \emph{cycle}.
A \emph{fast path} is a bounded optimistic attempt on a
claimed ticket, whereas a \emph{slow path} publishes the operation as a request that other threads can complete on the owner's behalf (\emph{helping}). 
A queue is \emph{linearizable} if every operation appears to take effect atomically between invocation and response~\cite{linearizability} in FIFO order. An operation is \emph{blocking} if a delayed thread can prevent others from completing, \emph{lock-free} if some operation always completes in a finite number of system steps, and \emph{wait-free} if every operation completes in a bounded number of its \emph{own} steps~\cite{herlihy-book}. The following proofs are stated for sequentially consistent executions of the pseudocode. The implementation uses relaxed device-scope atomics and places a fence or an acquire load wherever a proof orders two accesses to different words (\S\ref{sec:setup}).

\subsection{Packed-Ring Queue: Lock-Free (PRQ-LF)}
\label{sec:prqlf}

\paragraph{\textbf{Overview}} PRQ-LF keeps the ring properties that sCQ
established~\cite{scq}: $2n$ slots for capacity $n$, monotonic counters, and a threshold-based emptiness test. It changes how the queue touches memory, which is where GPU queue performance is decided. Three design choices distinguish PRQ-LF from sCQ. First, all three contended words are updated warp-cooperatively. Ticket reservation on \texttt{Head} and \texttt{Tail} goes through a batched FAA we call \textsc{WarpFAA}, the emptiness threshold through a batched decrement we call \textsc{WarpSub}, and the threshold reset is written by one lane on behalf of the warp. Second, each slot's full state fits in one 64-bit word, so every slot update is a single atomic operation. Third, \texttt{Head}, \texttt{Tail}, and \texttt{Threshold} each occupy a private cache line, so producer-side atomic traffic does not invalidate the line consumers poll. These changes alter how tickets and threshold updates are obtained, not the queue state changes the updates produce, so the
sCQ correctness arguments carry over once that equivalence is
established (Lemma~\ref{lem:warpbatch}).

\[
    \mathit{Entry} = \langle \mathit{Cycle}_{15}, \mathit{Safe}_1,
    \mathit{Val}_{48} \rangle.
\]

\paragraph{\textbf{Data structure}} Each slot is one 64-bit word $\mathit{Entry}$, where $\mathit{Val}$ is the payload or one of two reserved values, $\bot$ for a slot never filled this cycle and $\bot_c$ for a consumed slot. sCQ keeps payloads in a separate array and moves array indices through the ring. PRQ-LF stores the payload in the slot itself, which removes the second ring that manages free indices and the dependent memory access that indirection costs on every operation. The queue state is completed by \texttt{Head}, \texttt{Tail}, and a signed \texttt{Threshold} with maximum $\Theta = 4n{-}1$. The reset value is one $n$ larger than sCQ's $3n{-}1$ because up to n concurrent in-flight enqueuers can transiently push \texttt{Tail} to $3n$ ahead of \texttt{Head}. \textsc{Consume} atomically marks a slot's $\mathit{Val}$ as $\bot_c$ without changing the other fields.

\paragraph{\textbf{Warp-batched updates}} A na\"ive FAA queue performs one FAA per active thread on a counter word, needlessly concentrating contention. \textsc{WarpFAA} (Algorithm~\ref{alg:prqlf}) instead has one \emph{leader} lane FAA the counter by active-lane count and broadcasts the returned base ticket. Each lane then adds its \emph{rank} within the active mask. The warp obtains a consecutive block of tickets from one global atomic. \textsc{WarpSub} applies the same pattern to the threshold: the leader subtracts the active-lane count once, and each lane derives the value it would have observed as the returned value minus its rank. Warp aggregation itself is an established GPU optimization~\cite{Zhang14,bacq}.

\begin{algorithm}[t]
\caption{PRQ-LF ring attempts. $\mathrm{cmp}$ is modular cycle comparison (Lemma~\ref{lem:cycle}). A failed CAS refreshes $E$. Caches and slot-permutation details are omitted.}
\label{alg:prqlf}
\SetKwProg{Fn}{function}{:}{end}
\Fn{\textsc{WarpFAA}$(C)$}{
  $m \gets \textsc{Ballot}(\mathit{active})$;
  \lIf{$m = 0$}{\Return $\bot$}
  \lIf{$\mathit{lane} = \textsc{Leader}(m)$}{
    $\mathit{base} \gets \textsc{FAA}(C, \textsc{Popc}(m))$}
  $\mathit{base} \gets \textsc{Shuffle}(\mathit{base},
    \textsc{Leader}(m))$\;
  \Return $\mathit{base} + \textsc{Rank}(m, \mathit{lane})$\;
}
\Fn{\textsc{TryEnq}$(x)$}{
  \lIf{$\texttt{Tail} - \texttt{Head} \ge 2n$}{\Return 
  \textsc{Saturated}
    \tcp*[f]{pre-check}}
  $t \gets \textsc{WarpFAA}(\texttt{Tail})$;
  $j \gets \slotfn(t)$; $c \gets \cyclefn(t)$;
  $E \gets \mathit{Entry}[j]$\;
  \While{$\mathrm{cmp}(E.\mathit{Cycle}, c) < 0$ \KwAnd
      $(E.\mathit{Safe} \vee \texttt{Head} \le t)$ \KwAnd
      $E.\mathit{Val} \in \{\bot, \bot_c\}$}{
    \lIf{$\textsc{CAS}(\mathit{Entry}[j], E, \langle c, 1, x\rangle)$
        \emph{succeeds}}{
      reset $\texttt{Threshold} \gets \Theta$ (leader only);
      \Return \textsc{Success}}
  }
  \lIf{$(t+1) - \texttt{Head} \ge 2n$}{\Return \textsc{Saturated}
    }
  \Return \textsc{Retry}\;
}

\Fn{\textsc{TryDeq}$()$}{
  \lIf{$\texttt{Tail} \leq \texttt{Head}$ or $\texttt{Threshold} < 0$}{\Return \textsc{Empty}}
  $h \gets \textsc{WarpFAA}(\texttt{Head})$;
  $j \gets \slotfn(h)$; $c \gets \cyclefn(h)$;
  $E \gets \mathit{Entry}[j]$\;

  \If{$\mathrm{cmp}(E.\mathit{Cycle}, c) = 0$ \KwAnd
      $E.\mathit{Val} \notin \{\bot, \bot_c\}$}{
    $\textsc{Consume}(\mathit{Entry}[j])$;
    \Return $E.\mathit{Val}$\;
  }

\If{$\mathrm{cmp}(E.\mathit{Cycle}, c) < 0$}{
    $D \gets E.\mathit{Val} \in \{\bot, \bot_c\}$ ? $\langle c, E.\mathit{Safe}, \bot\rangle$ : $\langle E.\mathit{Cycle}, 0, E.\mathit{Val}\rangle$\;
    \lIf{$\textsc{CAS}(\mathit{Entry}[j], E, D)$ \emph{fails}}{
      \textbf{goto} line 16}
}

  \If{$\texttt{Tail} \le h + 1$}{
    catch \texttt{Tail} up to $h{+}1$ (leader only, bounded);
    $\textsc{WarpSub}(\texttt{Threshold})$;
    \Return \textsc{Empty}\;
  }
  \lIf{$\textsc{WarpSub}(\texttt{Threshold}) \le 0$}{
    \Return \textsc{Empty}}
  \Return \textsc{Retry}\;
}

\end{algorithm}

\paragraph{\textbf{Boundaries and abandoned tickets}} A fast-path attempt whose slot proves unusable returns \textsc{Retry}, and the operation's next attempt claims a fresh ticket. We call the unused ticket \emph{abandoned}. FAA is never undone, so \texttt{Head} and \texttt{Tail} order ticket attempts rather than directly encoding the number of live payloads. A dequeuer whose ticket maps to a never-filled slot installs $\langle c, E.\mathit{Safe}, \bot \rangle$, barring straggling enqueuers of older cycles from that slot exactly as in sCQ; empty responses remain governed by Threshold~\cite{scq}. If this CAS fails, an enqueuer may have just published into the slot, so
the dequeuer rereads the entry and tries the same ticket again; abandoning
the ticket would leave the published item with no dequeuer to remove it.

The two enqueue boundary checks in Algorithm~\ref{alg:prqlf} implement the imprecise full test of sCQ. Up to $k \le n$ concurrent enqueuers may advance Tail without a corresponding completed publication, so a failing
$2n$-distance check guarantees at least $2n - k \ge n$ queued elements, although the exact occupancy can be larger. PRQ-LF therefore has
guaranteed capacity $n$ and hard capacity $2n$: above $n$ elements an
enqueue may either succeed or return Saturated, depending on the
concurrent execution. A Saturated enqueue leaves the queue unchanged, and
the caller still holds its item. Applications that need a hard bound of
$n$ items can place payloads behind sCQ's free-index ring, as PRQ-WF does.

\begin{lemma}[\textbf{Warp batching preserves ticket and threshold order}]
\label{lem:warpbatch}

Within one \textsc{WarpFAA} call the active mask is fixed: the leader's single FAA by the active-lane count reserves a contiguous ticket interval, and distinct per-lane ranks partition it, so returned tickets are pairwise distinct and consecutive. Atomic FAA keeps the counter monotone across calls, so \textsc{WarpFAA} produces exactly the per-thread total ticket order and changes only how it is obtained. \textsc{WarpSub} is identical with per-lane values $\mathit{old}-\mathit{rank}$, so every threshold test observes a value occurring in the per-thread execution; the leader-only reset writes the same constant as a per-thread reset would.

\end{lemma}

\begin{lemma}[\textbf{Reduced-width cycle tags}]
\label{lem:cycle}
Let cycle tags (and, in PRQ-WF, notes) have $w$ bits, and let $R = 2^{w}$.
Tags are compared modulo $R$: $x$ is newer than $y$ when
$0 < (x-y) \bmod R < R/2$. Suppose that while any fast-path attempt or
slow-path round is in progress, neither Head nor Tail advances by $Rn/2$
or more tickets. Then every comparison between a stored tag and the cycle
of a ticket agrees with the comparison of their true cycles, and a CAS
whose expected entry was read in one generation of a slot cannot succeed
in a later generation. PRQ-LF uses $w = 15$ and PRQ-WF uses $w = 8$.
\end{lemma}

\begin{proof}
Head and Tail stay within $3n$ of each other: the saturation checks keep
Tail less than $2n + k$ ahead of Head, and a dequeuer that overtakes Tail
catches it up. A slot moves to a newer generation whenever a ticket for
it is served, so a stored generation falls behind the counters only while
an attempt on that slot is still in progress, for example an item waiting
for its pending dequeuer. Every ticket in use likewise belongs to an
attempt in progress. By the assumption, all of these generations lie
within fewer than $R/2$ cycles of each other, and integers less than
$R/2$ apart are ordered correctly by their residues modulo $R$. A later
generation of a slot therefore carries a different tag than the expected
entry, so a stale CAS fails.
\end{proof}

A tag must not wrap while an operation that reads it is still in progress, and sCQ assumes similar about its counters~\cite{scq}. With $n = 2^{16}$, the condition allows $2^{30}$ reservations on a counter during a single attempt in PRQ-LF and $2^{23}$ in PRQ-WF.

\begin{theorem}[\textbf{PRQ-LF is linearizable and FIFO}]
\label{thm:lf-lin}
If at most $n$ threads participate and the condition of
Lemma~\ref{lem:cycle} holds, PRQ-LF is a linearizable FIFO queue that
accepts at least $n$ items before it may return Saturated.
\end{theorem}
\begin{proof}
By Lemma~\ref{lem:warpbatch}, each WarpFAA and WarpSub can be replaced by the
equivalent sequence of per-thread updates, and by
Lemma~\ref{lem:cycle}, every tag comparison gives the same answer as a
comparison of true cycles. The resulting execution is an execution of
sCQ's direct-payload queue: each slot update is a single 64-bit CAS or
Consume that makes the same state change as in sCQ, and the counters and
threshold follow sCQ's rules. sCQ's correctness argument~\cite{scq}
therefore gives a FIFO linearization. A Saturated response changes no
state and is permitted by the capacity specification. We do not fix the
linearization point of an enqueue at its publishing CAS, since two
overlapping enqueues can publish out of ticket order; the linearization
follows ticket order.
\end{proof}

\begin{theorem}[\textbf{PRQ-LF is lock-free}]
\label{thm:lf-lf}
Under the assumptions of Theorem~\ref{thm:lf-lin}, PRQ-LF is lock-free.
\end{theorem}

\begin{proof}
By Lemma~\ref{lem:warpbatch}, we can consider only per-thread execution.
Suppose, for contradiction, that some infinite execution completes only
finitely many operations. Then only finitely many operations are invoked,
and from some point on every step is taken by a fixed, finite set of
pending operations, none of which completes.

Every successful publication completes an enqueue, so after that point no
publication occurs, and hence no reset of Threshold. A dequeue attempt
that does not return decrements Threshold before retrying, and a dequeue
returns Empty once it observes a value at or below zero. Threshold never
exceeds $4n-1$ and no longer increases, so the pending dequeues start only
finitely many further attempts. Each attempt changes at most one slot,
and a CAS fails only after another operation has changed the same slot,
so these attempts also take finitely many steps. Eventually only enqueues
take steps, and Head no longer changes.

From then on no slot changes, so every enqueue attempt ends without
retrying its CAS, and the next attempt reserves a new ticket from Tail.
Tail therefore grows without bound while Head stays fixed, and some
enqueue eventually finds $(t+1) - \mathit{Head} \ge 2n$ and returns
Saturated. This completes an operation, a contradiction. Hence every
infinite execution completes infinitely many operations, and PRQ-LF is
lock-free.
\end{proof}

\subsection{Packed-Ring Queue: Wait-Free (PRQ-WF)}
\label{sec:prqwf}
\paragraph{\textbf{Design goal}} PRQ-WF adds a cooperative slow path to the same $2n$-slot ring.
After at most $P_e$ (enqueue) or $P_d$ (dequeue) fast-path attempts, which are compile-time patience constants, an operation publishes a fixed per-thread request, and each participant executes one helping round every $D$ of its own operations. Relative to wCQ~\cite{wcq}, PRQ-WF replaces CAS2 with native 64-bit state, and replaces unconditional
check-one-peer helping with gated deterministic discovery plus an accelerating bitmap.

\paragraph{\textbf{Data structure}}

A PRQ-WF slot carries helping state alongside its payload, so payloads live in a separate array indexed by a 30-bit \emph{index} field, with a second ring managing free indices as in sCQ~\cite{scq}. This bounds occupancy by $n$, since an enqueue must first acquire one of $n$ free indices. Each slot is one 64-bit word:
\[
\mathit{Entry} = \langle \mathit{Note}_8,\ \mathit{Cycle}_8,\
\mathit{Safe}_1,\ \mathit{Enq}_1,\ \mathit{ProdTid}_{16},\
\mathit{Index}_{30} \rangle .
\]

$\mathit{Note}$ records the newest cycle for which the slot has been ruled out (Lemma~\ref{lem:stale}). $\mathit{Enq}$ distinguishes ordinary entries from not-yet-finalized slow-path entries. PRQ-WF adds a field absent from wCQ: $\mathit{ProdTid}$, which is the record index of the slow-path enqueuer that installed an $\mathit{Enq}{=}0$ entry, written by the same CAS as the rest of the entry and ignored when $\mathit{Enq}{=}1$ (Lemma~\ref{lem:finalize}). The ring keeps the same threshold-based empty test as PRQ-LF with sCQ's reset value $3n{-}1$, which suffices here because at most $n$ indices circulate. The global \texttt{Head} and \texttt{Tail} each pack a counter and a thread
identifier into one word, $\langle \mathit{Counter}_{48}, \mathit{ThrIdx}_{16} \rangle$, where $\mathit{ThrIdx}$ names the thread whose slow-path increment is in flight, or a reserved null value. The queue is configured with a fixed span of $K$ request records, one per potential participant. The 16-bit $\mathit{ThrIdx}$ field with one reserved null value fixes $K = 2^{16}-1$. We write $k$ for the number of participants active in a launch, with $k \le K$. 
Each thread owns a fixed record containing the request fields, the initial and local head/tail values, the payload index for an enqueue, two sequence fields $\mathit{seq1}$ and $\mathit{seq2}$, one phase-2 record for each of the two counters, a cohort word $\langle \mathit{epoch}_{32}, \mathit{active}_{32} \rangle$ that counts the helpers currently working on the request and enforces the eager-helper cap below, and a private helping schedule, a countdown of $D$ and a round-robin cursor. The queue additionally holds a gate counter \texttt{PendActive}, a count of published requests that have not yet retired, and a two-level pending bitmap with one bit per participant, used only to locate
pending requests quickly.

\begin{lemma}[\textbf{Single-word shared-state atomicity}]
\label{lem:atomicity}
Each of the following shared words concurrently modified in PRQ-WF is updated with one 64-bit atomic operation: slot entries, global head and tail state, local head and tail state, the cohort word, the gate, and the bitmap words. Consequently, helpers never observe torn mixed states, such as a new cycle with an old index, a new counter with a stale helper identifier, or an entry whose $\mathit{ProdTid}$ disagrees with its payload. 

\end{lemma}

\begin{algorithm}[t]
\caption{PRQ-WF enqueue and dequeue. \textsc{WarpLoad} broadcasts to the warp. \textsc{Ctr} extracts the counter from a packed word.}
\label{alg:prqwf}
\SetKwProg{Fn}{function}{:}{end}
\Fn{\textsc{Enqueue}$(x)$}{
  \textsc{HelpThreads}()\;
  \For{$P_e$ \emph{attempts}}{
    \lIf{$\textsc{TryEnq}(x) = \textsc{Success}$}{\Return}
  }
  \textbf{publish} $r$: write fields; cohort $\gets \langle s, 0
    \rangle$; fence; $\mathit{seq2} \gets s$; fence;
    $\texttt{PendActive}{+}{+}$, set bitmap bits;
    $\mathit{pending} \gets 1$\;
  \While{$\textsc{SlowFAA}(\texttt{Tail}, r.\mathit{localTail},
      \varnothing)$}{
    \lIf{$\textsc{TryEnqSlow}(\textsc{Ctr}(r.\mathit{localTail}), x,
      r)$}{\textbf{break}}
  }
  \textbf{retire} $r$: $\mathit{seq1} \gets s{+}1$; fence; close cohort epoch;
    $\mathit{pending} \gets 0$; clear bitmap bit,
    $\texttt{PendActive}{-}{-}$\;
}
\Fn{\textsc{Dequeue}$()$}{
  \textsc{HelpThreads}()\;
  \lIf{$\texttt{Threshold} < 0$}{\Return \textsc{Empty}}
  $H_0 \gets \textsc{WarpLoad}(\texttt{Head})$;\quad
  $T_0 \gets \textsc{WarpLoad}(\texttt{Tail})$\;
  \lIf{$T_0 \le H_0$}{\Return \textsc{Empty}
    \tcp*[f]{Lemma~\ref{lem:precheck}}}
  \For{$P_d$ \emph{attempts}}{
    \lIf{$\textsc{TryDeq}()$ \emph{succeeds}}{\Return value}
  }
  \textbf{publish} $r$ as above (dequeue kind)\;
  \While{$\textsc{SlowFAA}(\texttt{Head}, r.\mathit{localHead},
      \&\texttt{Threshold})$}{
    \lIf{$\textsc{TryDeqSlow}(\textsc{Ctr}(r.\mathit{localHead}),
      r)$}{\textbf{break}}
  }
  \textbf{retire} $r$ as above; gather result from ticket
    $\textsc{Ctr}(r.\mathit{localHead})$; \Return value or
    \textsc{Empty}\;
}
\end{algorithm}

\paragraph{\textbf{Publication and retirement}} A slow-path request is published in a fixed order (Algorithm~\ref{alg:prqwf}). The owner writes the request fields and opens its cohort word to $\langle s, 0 \rangle$ for its current sequence value $s$. It then writes $\mathit{seq2} \gets s$, increments the gate, sets its bitmap bits, and finally sets $\mathit{pending}$, with fences between these steps. Helpers read $\mathit{seq2}$, then the fields, then $\mathit{seq1}$, and accept the request only when the sequence values match. The gate is opened before the request becomes visible, so a zero gate soundly proves absence (Lemma~\ref{lem:gate}), and the cohort epoch is opened before $\mathit{seq2}$, so a helper that validates the sequence values can join only this request's cohort. Retirement inverts the order. The owner advances $\mathit{seq1}$ and fences, closes the cohort epoch with a CAS that also resets the active count, clears $\mathit{pending}$ and its bitmap bit, and decrements the gate. The owner does not wait for the helpers it admitted. A helper that is still running when the record carries a new request detects the change through $\mathit{seq1}$ and stops (Lemma~\ref{lem:reuse}).

\paragraph{\textbf{Gated helping}} 
Every thread runs one helping round per $D$ of its own operations (Algorithm~\ref{alg:help}). A warp-batched gate load of zero ends the round immediately. We call this \emph{gated helping}, because it costs one load when nothing is pending. Otherwise the thread advances its round-robin cursor by one participant, on every open-gate round, unconditionally, and probes that peer's $\mathit{pending}$ flag, adopting the request if the flag is set. We call this a \emph{deterministic} adoption, and it is never refused. Only if the cursor's peer is not pending does the thread consult the bitmap, scanning from a per-thread rotated start and filtering to its helper group (a fixed identifier-based partition that thins concurrent joiners). We call this an \emph{eager} adoption, and it is refused once $J$ helpers have joined the request. Rotation and filtering spread concurrent helpers across victims and carry no proof obligation (Lemma~\ref{lem:adoption}). The victim's request is validated through the sequence fields, and the helper joins the cohort, runs the victim's slow path, and leaves the cohort when done.

\begin{algorithm}[t]
\caption{PRQ-WF helping round}
\label{alg:help}
\SetKwProg{Fn}{function}{:}{end}
\Fn{\textsc{HelpThreads}$()$}{
  \lIf{$\textsc{WarpLoad}(\texttt{PendActive}) = 0$}{\Return
    }
  $p \gets \mathit{cursor}$;\quad $\mathit{cursor} \gets (p+1) \bmod K$
    \;
  \eIf{$\mathit{Record}[p].\mathit{pending}$}{
    $v \gets p$;\quad $\mathit{eager} \gets \textbf{false}$
      \;
  }{
    $v \gets \textsc{BitmapPick}(\text{rotated start, group filter})$;
    $\mathit{eager} \gets \textbf{true}$\;
    \lIf{$v = \varnothing$}{\Return}
  }
  read $v$'s request via $\mathit{seq2}$, fields, $\mathit{seq1}$;
    \Return on mismatch\;
  \If{$\textsc{CohortJoin}(v, \mathit{seq}, \text{cap enforced iff }
      \mathit{eager})$}{
    run $v$'s slow driver;\quad $\textsc{CohortLeave}(v)$\;
  }
}
\end{algorithm}

\begin{algorithm}[t]
\caption{PRQ-WF cooperative slow-path increment. The caller validated the
request under sequence value $s$. Fences separate the phase-2 steps
(\S\ref{sec:setup}).}
\label{alg:slowfaa}
\SetKwProg{Fn}{function}{:}{end}
\Fn{\textsc{SlowFAA}$(G, L, \mathit{thld})$}{
  \While{\textbf{true}}{
    \lIf{$L$ \emph{has} $\mathit{FIN}$}{\Return \textbf{false}}
    read $G = \langle c, u \rangle$\;
    \If{$u \ne \textsc{null}$}{
      read $u$'s phase-2 record for $G$; re-read $G$ and the record's
      version\;
      \lIf{\emph{both unchanged}}{clear $\mathit{INC}$ on the recorded
        local word}
      clear $u$ from $G$\;
    }
    synchronize $L$ to $c$ using $\mathit{INC}$\;
    \lIf{$L$ \emph{was adopted} \textbf{and} $\mathit{seq1} \ne s$}{%
      \Return \textbf{false}}
    publish phase-2 record $(L, c)$ with a new version\;
    \If{$\textsc{CAS}(G, \langle c, \textsc{null}\rangle,
        \langle c{+}1, \mathit{ThrIdx}\rangle)$ \emph{succeeds}}{
      \lIf{$\mathit{thld} \ne \varnothing$}{
        $\textsc{FAA}(\mathit{thld}, -1)$}
      clear $\mathit{INC}$ on $L$;\quad clear $\mathit{ThrIdx}$ in $G$;
      \Return \textbf{true}\;
    }
  }
}
\end{algorithm}

\paragraph{\textbf{Cooperative slow path}} Algorithm~\ref{alg:slowfaa} shows the cooperative slow-path increment (\textsc{SlowFAA}) used by both enqueue and dequeue to advance the global counter exactly once per round. A request's local $\mathit{Head}$ or $\mathit{Tail}$ word holds the ticket the request currently works on and two flags: $\mathit{INC}$, set while an increment for the current round is in flight, and $\mathit{FIN}$, set once the request has been resolved. Before its global CAS, a participant writes the pair (local word, counter value) and a new version number into its phase-2 record for that counter, so that any helper can clear $\mathit{INC}$ for an increment whose winner has stalled. A helper applies a phase-2 record only after re-reading the global word and the record's version and finding both unchanged, which proves that the record belongs to an increment whose global CAS succeeded (Lemma~\ref{lem:round}). When a participant's CAS on the local word fails, another participant has already moved the request to a newer ticket. The participant then \emph{adopts} that ticket, but first re-reads the owner's $\mathit{seq1}$ and stops if the request it validated has retired (Lemma~\ref{lem:reuse}). \textsc{TryEnqSlow} and \textsc{TryDeqSlow} operate on the slot selected by the current logical ticket. A slow enqueue either installs the current-cycle entry, packing its own record index as $\mathit{ProdTid}$ in the same CAS, or records via $\mathit{Note}$ that a stale slot is no longer a valid candidate for this request. A slow dequeue mirrors this: it completes on the matching current-cycle slot or updates slot state so later helpers reach the same reuse decision. When \texttt{Threshold} is negative, a slow dequeue tries to set $\mathit{FIN}$ for an empty result at its ticket and returns to \textsc{SlowFAA} whether or not that CAS succeeds. \textsc{SlowFAA} then stops on $\mathit{FIN}$ or adopts the newer ticket, so a request is never resolved at a ticket it no longer holds.

\begin{lemma}[\textbf{Gate soundness}]
\label{lem:gate}
\texttt{PendActive} is incremented by the owner before $\mathit{pending}$ becomes visible and decremented after $\mathit{pending}$ is cleared, exactly once per request. Hence $\texttt{PendActive} = 0$ implies that no published, unretired request exists at the instant of the read, so a helping round that observes zero skips only rounds with no one to help. A request published after the read is reached by that thread's next round, at most $D$ of its operations later. 
\end{lemma}

\begin{lemma}[\textbf{Deterministic adoption}]
\label{lem:adoption}
While a request remains pending, its owner holds the gate positive (Lemma~\ref{lem:gate}), so every helping round of every continuing participant advances its round-robin cursor (Algorithm~\ref{alg:help}). Each participant therefore probes the pending record within $K$ helping rounds while the gate remains positive, that is, within $D \cdot K$ of its own operations. The cohort cap applies only to bitmap-selected adoptions and never refuses a cursor-selected one, so capping removes only redundant helpers. The bitmap is only an accelerator. The cursor guarantees discovery, and sequence validation rejects retired requests. This supplies the all-participant convergence premise of the wCQ completion argument~\cite{wcq}.
\end{lemma}

\begin{lemma}[\textbf{Safe record reuse}]
\label{lem:reuse}
A helper accepts a request only after reading $\mathit{seq2} = \mathit{seq1} = s$, and it captures the request's fields between the two reads. A request's first local ticket is the ticket of its last failed fast-path attempt, a fresh FAA value ($P_e, P_d \ge 2$ guarantees that attempt exists), and local tickets only increase. Every local value of a request therefore differs from every local value of an earlier request in the same record, and a delayed helper's CAS on the local word, which expects a value of the request it validated, fails once the record carries a later request. On that failure the helper fences and re-reads $\mathit{seq1}$. The owner advances $\mathit{seq1}$ and fences before it writes a later request's local word, so the helper finds $\mathit{seq1} \ne s$ and stops. A helper therefore never applies fields captured from one request to a later one, and the owner reuses its record without waiting for any helper. A slow driver may return after recognizing this request's installed entry before the participant that installed it has set $\mathit{Enq}$. That participant's remaining CASes expect this request's ticket or this exact entry word, so after reuse they fail without effect, and the entry's consumer finalizes it without that participant (Lemma~\ref{lem:finalize}). Retirement takes a constant number of the owner's own steps plus failed cohort CASes, each charged to a distinct join, leave, or close transition of the same request; a stale helper's leave is an epoch-checked no-op.

\end{lemma}

\begin{lemma}[\textbf{$O(1)$ producer finalization}]
\label{lem:finalize}
The only transition that creates an $\mathit{Enq}{=}0$ entry is \textsc{TryEnqSlow}, which writes its own record index as $\mathit{ProdTid}$ in the same installing CAS (Lemma~\ref{lem:atomicity} rules out tearing). Fast-path and initial entries carry $\mathit{Enq}{=}1$. A consumer of an $\mathit{Enq}{=}0$ entry uses $\mathit{ProdTid}$ to locate the producer's record and attempts one CAS to set $\mathit{FIN}$, matching the full ticket with $\mathit{INC}$ and $\mathit{FIN}$ clear. Under the no-wrap assumption, an already-finalized request or a later enqueue request in the same record cannot match this expected word. 
This replaces wCQ's producer search over per-thread records, $O(k)$ work per such consume, with $O(1)$.
\end{lemma}

\begin{lemma}[\textbf{One cooperative increment per round}]
\label{lem:round}
For any published slow-path request, all helpers race on the same global transition $\langle c, \textsc{null} \rangle \to \langle c{+}1, \mathit{tid} \rangle$. At most one CAS can succeed. A participant clears $\mathit{INC}$ through another participant's phase-2 record only if, after reading the record, the global word still names that participant and the record's version is unchanged. A participant removes its identifier from the global word before it writes a new version, so these two checks prove that the version read belongs to an increment whose global CAS succeeded. One record per counter keeps an identifier left in \texttt{Tail} from naming the record in use for \texttt{Head}, and vice versa. The $\mathit{INC}$ bit prevents duplicate local increments for the same round, and the $\mathit{FIN}$ bit terminates further rounds once the request has been resolved. Therefore each logical slow-path round advances \texttt{Head} or \texttt{Tail} exactly once. For dequeue rounds, the threshold is decremented at most once, by the winning helper, scalar and unbatched. Warp-batched arithmetic is confined to the fast path, where it aggregates independent tickets (Lemma~\ref{lem:warpbatch}) rather than a shared round.
\end{lemma}

\begin{lemma}[\textbf{Stale-slot exclusion}]
\label{lem:stale}
If a helper determines that an old-cycle slot is not reusable for the
current request, it advances $\mathit{Note}$ to the current cycle. Any
later helper for that same request cycle observes the updated note and
skips the same slot. Thus helpers never reconsider a slot that has been ruled out. 

\end{lemma}

\begin{lemma}[\textbf{Linearizable empty pre-check}]
\label{lem:precheck}
Dequeue reads $H_0 = \mathit{Head}$ and then $T_0 = \mathit{Tail}$
(Algorithm~\ref{alg:prqwf}), each with a single warp-batched load, and the
\texttt{Tail} read is ordered after the \texttt{Head} read. If $T_0 \le H_0$,
the Empty response can be linearized at the read of $T_0$.
\end{lemma}

\begin{proof}
Both counters only increase, so when $T_0$ is read, every enqueue ticket
issued so far lies below $T_0 \le H_0$, and a dequeue ticket has already
been issued for each of them. An item published at such a ticket is
removed by the dequeuer holding that ticket, which may still be pending.
We linearize those pending removals before the read of $T_0$, and any
enqueue that publishes afterwards after it. The abstract queue is then
empty at the read of $T_0$. Because the check follows HelpThreads, an
Empty return still advances the helping schedule of Lemma~\ref{lem:adoption}.
\end{proof}

\begin{theorem}[\textbf{PRQ-WF is linearizable and FIFO}]
\label{thm:wf-lin}
If at most $n$ threads participate and the condition of
Lemma~\ref{lem:cycle} holds, PRQ-WF is a linearizable FIFO queue.
\end{theorem}

\begin{proof}
Lemma~\ref{lem:atomicity} makes every update of a shared word atomic, and
Lemma~\ref{lem:cycle} lets us compare cycles as true integers. By
Lemma~\ref{lem:round}, each cooperative round advances Head or Tail
exactly once, so slow-path and fast-path tickets form a single ticket
order, as in wCQ. Lemmas~\ref{lem:reuse}, \ref{lem:finalize},
and~\ref{lem:stale} ensure that helpers act only on the request they
validated, finalize the correct producer, and never revisit a slot that
has been ruled out. Each index ring therefore makes the same state
changes as wCQ, and wCQ's argument~\cite{wcq} gives a FIFO linearization;
Empty responses from the pre-check are placed by
Lemma~\ref{lem:precheck}. Payloads travel with their indices through the
two rings as in sCQ, and an enqueue that finds no free index returns
without changing the queue, so the payload queue is linearizable as well.
\end{proof}

\begin{theorem}[\textbf{PRQ-WF is wait-free}]
\label{thm:wf-wf}
Assume that participants have unique identifiers
from the record span, with $k \le \min(K, n)$; that 64-bit device-scope
atomics are single-copy atomic with the fences listed in \S\ref{sec:setup}; that the
condition of Lemma~\ref{lem:cycle} holds; and that no counter or epoch
wraps during a live observation after a quiescent initialization. Then
PRQ-WF is wait-free: every operation completes within a number of its own
steps that depends only on $P_e$, $P_d$, $D$, $K$, and $n$.
\end{theorem}

\begin{proof}
The fast path is bounded by the patience constants, and each attempt is
itself bounded because a slot admits a constant number of state changes
per cycle. After that, the operation publishes a fixed-size request
record. By Lemmas~\ref{lem:gate} and~\ref{lem:adoption}, each continuing participant reaches that request's record within $K$ helping rounds, and deterministic adoption bypasses the eager-helper cap. Together with the slow-path transitions preserved by Lemmas~\ref{lem:finalize}--\ref{lem:stale}, this supplies the premises for wCQ's convergence argument~\cite{wcq}, including participants that first finish requests they adopted earlier. From then on, competing fast-path reservations on the contended counter stop, because converged participants execute the victim's rounds instead of their own, so the round CAS of Lemma~\ref{lem:round} cannot be starved. Each round performs exactly one logical increment of the global counter, with per-round slot work bounded by Lemma~\ref{lem:stale}. Once $\mathit{FIN}$ is set, the owner's \textsc{SlowFAA} returns at its next check, and retirement takes a bounded number of the owner's own steps without waiting for any helper (Lemma~\ref{lem:reuse}). Every phase is bounded in the owner's own steps, so every operation completes in a bounded number of its own steps.
\end{proof}

The bound holds for every schedule of the participating threads, including one that never runs a stalled warp again, because no step of PRQ-WF waits for another thread. It therefore needs neither residency of the whole launch nor a fairness bound on the warp scheduler, which GPU vendors do not document~\cite{sorensen-concur18}.

\section{Experimental Methodology}
\vspace{-2pt}
\label{sec:method}

We evaluate the queue under two micro-benchmark workloads introduced in \S~\ref{sec:intro}, the closed-loop kernel and the open-loop kernel, and under queue-backed level-synchronous BFS. We then profile the micro-benchmarks to explain their throughput trends.

\subsection{Experimental Setup}
\label{sec:setup}
\paragraph{Devices and software}

We use two NVIDIA GPUs with 32-thread warps, an A100 (Ampere, 108 SMs, 80\,GB) and a GH200 (Hopper, 132 SMs, 96\,GB), and two AMD GPUs with 64-thread wavefronts, an MI210 (CDNA2, 104 CUs, 64\,GB) and an MI300A (CDNA3, 228 CUs, 128\,GB unified memory). We compile C++17 with the same host compiler family, GCC 12.2, on all four systems. AMD binaries are compiled with HIP with \texttt{gfx90a} or \texttt{gfx942} under ROCm 7.1.1, and NVIDIA binaries with CUDA 13.0 for \texttt{sm\_80} or \texttt{sm\_90} target.

\paragraph{Queue integration and memory ordering}

Shared queue state resides in device global memory and is modified with
native device-scope 32- and 64-bit atomics, and both queues observe
control state with relaxed device-scope loads. PRQ-LF stores each payload
inside its slot word, so an item is published and consumed by single
atomic operations. PRQ-WF keeps payloads in a separate array: a producer
writes the payload and issues a device-wide fence before publishing its
index, and a consumer issues a fence after obtaining the index, reads the
payload, and fences again before returning the index to the free ring.
The PRQ-WF slow path places a device-wide fence between the writes of a
phase-2 record and the global CAS that names it; between a helper's read of
the global word and its reads of the record; between clearing
$\mathit{INC}$ and clearing the identifier from the global word; and between
observing that cleared identifier and the next CAS on the local word. An
owner fences after advancing $\mathit{seq1}$, and a helper fences before
re-reading it (Lemma~\ref{lem:reuse}). The empty pre-check orders its
\texttt{Tail} read after its \texttt{Head} read with an acquire load or a
fence (Lemma~\ref{lem:precheck}). Each baseline retains its own synchronization paths. SFQ uses its blocking calls in the closed-loop kernel and its non-blocking calls in the open-loop kernel, where blocking on an empty or full queue could keep the kernel from completing. PRQ-WF uses 16-bit thread identifiers, and its patience constants $P_e/P_d$ (\S~\ref{sec:prqwf}) are 2048/8192 on MI210, 1024/2048 on MI300A, and 4096/16384 on A100 and GH200.

\subsection{Queue Throughput Benchmarks}
\label{sec:method_tp}
\paragraph{\textbf{Closed-loop kernel}} Every active thread repeats one enqueue followed by one dequeue, 64 pairs per thread per launch. Consumed work is replaced by the same thread, so the queue experiences the access pattern of regenerative workloads such as frontier expansion (\S~\ref{sec:bfs}), with enqueue and dequeue demand coupled at every thread.

\paragraph{\textbf{Open-loop kernel}} Each thread holds one dedicated role for the entire test. A producer performs 64 enqueue attempts and a consumer performs 64 dequeue attempts. Producer shares of 25\%, 50\%, and 75\% drive the queue into the draining (near-empty), rate-matched, and filling (near-full) conditions. The rate-matched condition is not the
closed-loop kernel. Half of its threads only enqueue and half only dequeue, so producers and consumers can run ahead of or fall behind each other even though their requested operation counts are equal. The open-loop kernel therefore measures behavior under sustained imbalance rather than the cost of an enqueue/dequeue pair.

\paragraph{\textbf{Metric and procedure}}
A completed enqueue counts only when it publishes an item and a completed dequeue only when it removes and returns an item. A failed enqueue and an empty dequeue are attempts, but are not queue progress. For a measured interval,

\begin{align}
\label{eq:1}
S &= N_{\mathrm{enq\_complete}} + N_{\mathrm{deq\_complete}}, \\
\text{Throughput} &= \frac{S}{t_{\mathrm{kernel}}}\;[\mathrm{ops/s}].
\end{align}

We sweep $T \in \{2^9,2^{10},\ldots,2^{15}\}$ active threads with 256
threads per block. Each configuration first accumulates at least 100~ms of untimed kernel execution and then at least 500~ms of GPU-event kernel time. Allocation, initialization, queue reset, and host work are outside the measured interval. Queue order is rotated across repetitions. We report the median of 10 repetitions on each tested GPU platform.

\subsection{Queue-Backed Level-Synchronous BFS}
\label{sec:bfs}
BFS tests the consume-and-replace pattern when the input graph determines the amount of generated work. Graphs are stored in compressed sparse row (CSR) format. One multiple-producer, multiple-consumer FIFO holds both the current and next frontiers, with each item tagged by its BFS level. At a level, the kernel issues exactly the known number of current-frontier dequeues. A dequeued vertex scans its outgoing adjacency list, and the thread that first marks an unvisited neighbor enqueues that neighbor for the next level. The host then reads the next-frontier count and launches the following level. Execution ends when that count is zero.

We evaluate BQ, PRQ-LF, and PRQ-WF on \texttt{delaunay\_n24} and
\texttt{road\_usa} from the SuiteSparse Matrix Collection~\cite{suitesparse}.
The inputs are undirected graphs. Their fixed queue capacities are 65,536 and 32,768 elements, respectively, chosen as the next power of two at or above the maximum current-plus-next frontier measured by a CPU reference for each source. PRQ-LF's capacity is its guaranteed capacity $n$ ($2n$ slots). We use ten source vertices per graph. For every source, queue, GPU, and thread count, we perform two untimed warmups followed by ten
measured traversals. We use the throughput thread sweep and a block size of
256. Traversal time includes frontier seeding, all level kernels, the
level-by-level frontier-count transfers, and their required synchronization;
graph loading, queue allocation, and initial state construction are excluded.

Every run is validated against a CPU reference (reached set, levels, full distance vector) and rejected on any invalid payload, level-order violation, capacity overflow, leftover item, or enqueue/dequeue imbalance. SFQ is excluded because its supported participant bound does not cover the complete thread sweep, and BWD is excluded because it does not provide FIFO order.

\subsection{Profiling Setup}
\label{sec:profiling}
We profile the closed-loop and open-loop kernels with the same queue capacity, block size, thread sweep, and per-launch work described above. Each profiled configuration contains one kernel launch and is repeated three times, using rocprofv3 on AMD (versions 1.0.0 on MI210 and MI300A) and Nsight Compute on NVIDIA (2025.1.1 on A100, 2025.3.0 on GH200).

Hardware counters cover the complete kernel while the queue reports completed operations, so WAIT/op divides the kernel-wide waiting counter by $S = N_{\mathrm{enq,complete}} + N_{\mathrm{deq,complete}}$ from the same launch; atomic transactions/issue uses only atomic hardware events in both numerator and denominator, so it includes failed attempts. Table~\ref{tab:profile-metrics} gives exact vendor definitions of the counters. The paired counters answer the same two questions on both vendors — waiting per completed operation, and memory-system transactions per SIMT atomic issue — but are not identical physical units, so we compare queue ordering and thread-count trends within each GPU rather than absolute values across vendors. In the open-loop kernels, a full or empty return exits at the boundary check, so WAIT/op measures waiting per unit of progress.

\begin{table}[t]
  \centering
  \small
  \setlength{\tabcolsep}{2.5pt}
  \setlength{\arrayrulewidth}{0.4pt}
  \renewcommand{\arraystretch}{1.02}
  \caption{Cross-vendor profiling metrics.
  WAIT/op is normalized by $S$ (Equation~\ref{eq:1}).}
  \label{tab:profile-metrics}

  \begin{tabularx}{\columnwidth}{|
    >{\raggedright\arraybackslash}p{0.9cm}|
    >{\raggedright\arraybackslash}p{1.35cm}|
    >{\raggedright\arraybackslash}X|}
    \hline
    \textbf{Vendor} & \textbf{Plot term} & \textbf{Definition and unit} \\
    \hline

    AMD &
    WAIT/op &
    $\texttt{SQ\_WAIT\_ANY}/S$;
    wave quad-cycles per completed operation. \\

    AMD &
    Atomic transactions/issue &
    Atomic request sum divided by
    \texttt{TD\_ATOMIC\_WAVEFRONT\_sum};
    TCP-to-TCC requests per atomic wave issue. \\
    \hline

    NVIDIA &
    WAIT/op &
    $\texttt{smsp\_\_warps\_active.sum}
    (1-p_{\mathrm{selected}}-p_{\mathrm{not\_selected}})/S$;
    waiting warp-cycles per completed operation. \\

    NVIDIA &
    Atomic transactions/issue &
    L1TEX atomic sector sum divided by the corresponding atomic
    request sum; sectors per atomic request. \\
    \hline
  \end{tabularx}
\end{table}

\begin{figure*}[!t]
  \centering

  \begin{subcaptiongroup}
    \includegraphics[width=0.98\textwidth]
      {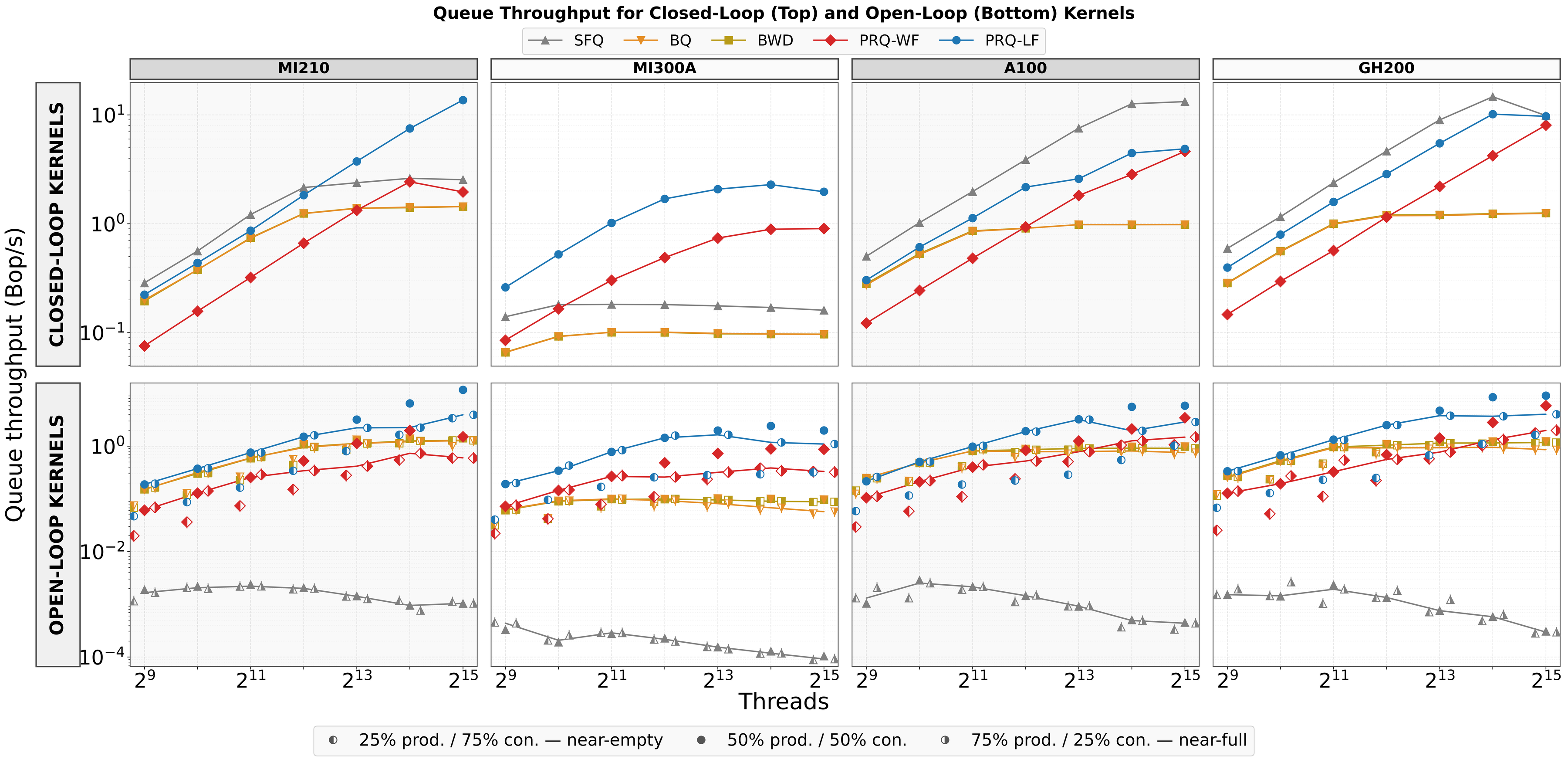}

    \phantomsubcaption\label{fig:closed-tp}
    \phantomsubcaption\label{fig:open-tp}
  \end{subcaptiongroup}

  \caption{Completed queue throughput across four GPUs.
  \textbf{(a)} Closed-loop kernel throughput, reported as the median of
  10 runs.
  \textbf{(b)} Open-loop kernel. At each thread count, the queue line is
  the median of the measured 25\%, 50\%, and 75\% producer conditions.
  The left-half, full, and right-half markers show the draining,
  rate-matched, and filling results, respectively.}
  \label{fig:throughput}
\end{figure*}

\subsection{Correctness Checks}
\label{sec:correctness}

\paragraph{\textbf{Linearizability via Porcupine}}
We log concurrent operations with fields \texttt{proc, op, arg, ret, call, end}, where \texttt{op}=0 denotes \textsc{Enq} and \texttt{op}=1 denotes \textsc{Deq}. We then feed the log to the standard FIFO model in Porcupine and check linearizability~\cite{linearizability,porcupine-linear}.\footnote{We follow Porcupine’s queue example; the model enqueues append to the state list and dequeues must return the head or report empty.}

\paragraph{\textbf{Device-side FIFO conformance}}
Independently, we run a GPU FIFO test: each producer thread emits tokens \texttt{tok = (tid << 32) | (seq+1)}; consumers dequeue until all tokens have been consumed. We verify (i) exactly once (no zeros, no $>1$ counts), (ii) no out-of-bounds tokens, and (iii) monotone sequence per-producer. For SFQ’s blocking variant, we cap ops/thread to avoid intentional blocking of producers.

\paragraph{\textbf{Result}} Porcupine reported every history as \emph{\textbf{linearizable}} for SFQ, BQ, PRQ-LF and PRQ-WF, and the device-side checker reported no duplicates or per-producer order violations. Together, these tests support linearizable FIFO semantics. The checked histories have up to 128 threads and 512 operations.

\section{Evaluation}
We report throughput under the closed-loop and open-loop kernels in billions of operations per second (Bop/s), queue-backed level-synchronous BFS, and a profiling analysis using the metrics of Section~\ref{sec:profiling}, across all four GPUs and against SFQ, BQ, and BWD.

\begin{figure*}[!t]
  \centering

  \begin{subcaptiongroup}
    \includegraphics[width=0.98\textwidth]
      {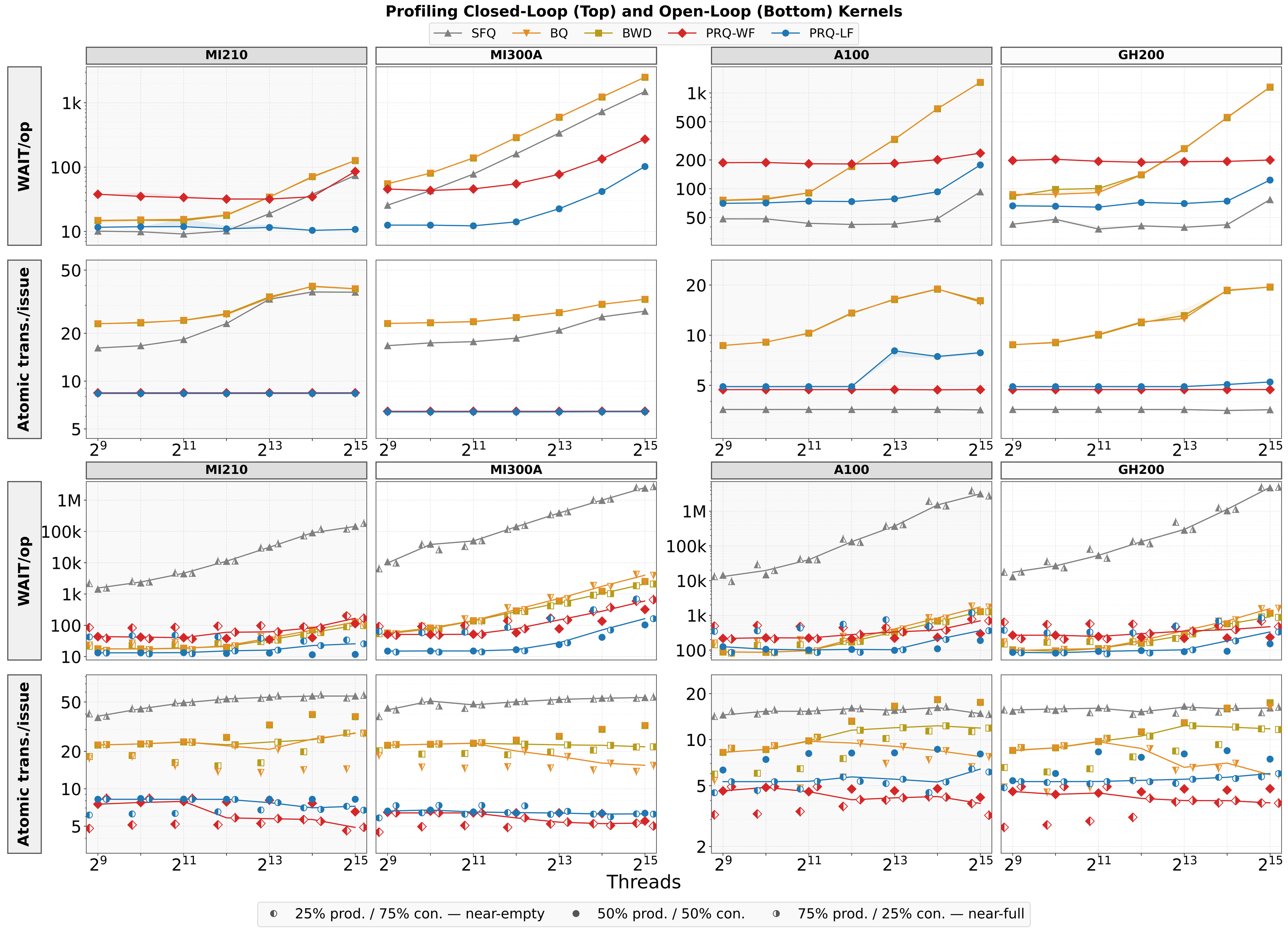}

    \phantomsubcaption\label{fig:closed-prof}
    \phantomsubcaption\label{fig:open-prof}
  \end{subcaptiongroup}

  \caption{WAIT/op and atomic transactions/issue for
  \textbf{(a)} the closed-loop and \textbf{(b)} the open-loop kernels,
  reported as the median of three profiling runs. Both metrics are defined
  in Section~\ref{sec:profiling} and Table~\ref{tab:profile-metrics}.
  Marker and line encoding follow Figure~\ref{fig:open-tp}.}
  \label{fig:profiling}
\end{figure*}

\subsection{Queue Throughput}
\label{sec:eval-observations}

\paragraph{\textbf{Closed-loop kernel.}}
Figure~\ref{fig:closed-tp} shows that the PRQ-LF is the fastest non-blocking queue on all four GPUs. PRQ-WF follows the same scaling trend at lower throughput, mostly visible on MI210 and MI300A. The gap reflects its index indirection, not helping: each PRQ-WF operation passes through two rings, and the slow path is rarely taken at the configured patience. At $2^{15}$ threads, PRQ-LF is 5.0--20.4$\times$ faster than BQ, and PRQ-WF is 1.4--9.3$\times$ faster. BQ and BWD throughput flatten together, showing that relaxing FIFO order does not remove the
shared broker bottleneck: they reserve and publish positions one operation at a time, whereas PRQ batches \texttt{Head} and \texttt{Tail} reservations at warp granularity. SFQ's blocking calls lead the closed loop only because a matched operation is always immediately available, while that matched supply disappears under dedicated thread roles.

\paragraph{\textbf{Open-loop kernel.}}
Figure~\ref{fig:open-tp} show that the draining, rate-matched, and filling conditions (25\%, 50\% and 75\% producer share respectively) stress empty detection, sustained producer--consumer separation, and full detection. PRQ-LF leads the rate-matched and filling conditions as concurrency grows on every GPU and consequently has the highest median across the three conditions for nearly the entire sweep. In the draining condition, BQ and BWD are often fastest at small and moderate thread counts, but at $2^{15}$ threads PRQ-WF leads both Broker Queues and slightly exceeds PRQ-LF on MI300A, A100, and GH200, while PRQ-LF remains strongest on MI210. PRQ-WF's advantage is specific to high-concurrency, near-empty pressure rather than imbalance in general.

\begin{table}[!t]
  \centering
  \caption{Median BFS traversal time in \textsc{ms}, lower is better, at each queue's best thread count selected by the lowest median time (ms). Bold marks the best per GPU and graph.}
  \label{tab:bfs-queue}
  \small
  \setlength{\tabcolsep}{2.2pt}
  \renewcommand{\arraystretch}{1.05}

  \begin{tabular}{@{}lrrr@{\hspace{4pt}}rrr@{}}
    \toprule
    & \multicolumn{3}{c}{\texttt{delaunay\_n24}}
    & \multicolumn{3}{c}{\texttt{road\_usa}} \\
    \cmidrule(lr){2-4}\cmidrule(lr){5-7}
    GPU
    & BQ & PRQ-WF & PRQ-LF
    & BQ & PRQ-WF & PRQ-LF \\
    \midrule
    MI210
    & 198.9 & 272.2 & \textbf{191.7}
    & 531.9 & 661.4 & \textbf{521.0} \\

    MI300A
    & 538.4 & 354.6 & \textbf{252.0}
    & 878.2 & 617.2 & \textbf{461.5} \\

    A100
    & 218.6 & 270.2 & \textbf{209.6}
    & 387.3 & 459.7 & \textbf{384.4} \\

    GH200
    & \textbf{159.3} & 185.9 & 290.3
    & \textbf{308.7} & 374.5 & 451.5 \\
    \bottomrule
  \end{tabular}
\end{table}

\subsection{Level-Synchronous BFS Performance}
Table~\ref{tab:bfs-queue} reports BFS time on \texttt{delaunay\_n24} and \texttt{road\_usa}. PRQ-LF is the fastest queue on three of the four GPUs on both graphs. On MI300A, the margin is large, 2.14$\times$ and 1.90$\times$ faster than BQ, and PRQ-WF is 1.52$\times$ and 1.42$\times$ faster. On MI210 and A100 the margin is small, between 0.8\% and 4.3\% faster. On GH200, BQ is the fastest. Each BFS level issues the known number of frontier dequeues, never polls an empty queue, and reaches only 21--49\% of queue capacity. Adjacency reads, distance-update atomics, and level synchronization therefore account for much of the traversal time, leaving less opportunity for higher queue throughput to reduce
end-to-end time. MI300A is the exception, where queue contention is high enough for the warp-cooperative paths to offset their additional queue-state updates.

\subsection{Profiling Analysis}
\label{sec:eval-profiling}

Figures~\ref{fig:closed-prof} and~\ref{fig:open-prof} explain why BQ and BWD flatten while PRQ-LF and PRQ-WF continue to improve.

\paragraph{\textbf{BQ and BWD}}
Both metrics increase with thread count on all four GPUs. At $2^{15}$ threads, BQ's \textsf{WAIT/op} is 7.3--24.5$\times$ that of PRQ-LF, and its atomic transactions/issue is 2.0--5.1$\times$ higher. BWD's relaxed ordering removes neither increase, so the Broker Queues pay
more waiting and more atomic traffic per unit of progress as concurrency grows. This matches their throughput flattening.

\paragraph{\textbf{PRQ-LF, PRQ-WF, and SFQ}}
Both PRQ designs keep atomic transactions/issue low across the thread sweep, which lets throughput continue rising after BQ and BWD flatten. PRQ-WF's higher closed-loop \textsf{WAIT/op} reflects its second ring and payload indirection. The profiles are split by condition. In the draining condition at $2^{15}$ threads, PRQ-WF has lower \textsf{WAIT/op} and lower atomic transactions/issue than PRQ-LF on MI300A, A100, and GH200, matching its throughput lead there; MI210 reverses this ordering. In the rate-matched and filling conditions, PRQ-LF has lower \textsf{WAIT/op} on every GPU. SFQ has low closed-loop metrics on the NVIDIA GPUs, but under the open-loop kernel fewer than 0.1\% of attempts complete at $2^{15}$ threads, so its very high
\textsf{WAIT/op} reports waiting per unit of progress rather than the cost of the few calls that complete.

\section{Future Work}
Our results point to three directions. First, PRQ-WF's slow path is rarely
taken, so its gap to PRQ-LF comes from the second index ring. A wait-free ring that stores payloads in the slot, as PRQ-LF does, would remove that indirection if the helping fields and a payload can share one
atomic word, for example on GPUs with 128-bit atomics. Shrinking the index field, which needs only 17 bits at $n = 2^{16}$, would also widen PRQ-WF's
8-bit cycle tags and relax the condition of Lemma~\ref{lem:cycle}. Second, throughput alone does not show what wait-freedom buys. Per-operation latency distributions under controlled delays would quantify the practical differences between blocking, lock-free, and wait-free queues on GPUs. Finally, irregular
GPU pipelines in which work arrives unevenly or in bursts would show when these properties improve application performance, and bounding how far dequeue reservations can run ahead of Tail may reduce wasted attempts on sparse workloads.

\section{Conclusion}
\label{sec:conclusion}

Packed-ring queues (PRQs) provide bounded-memory, linearizable GPU FIFOs with lock-free and wait-free progress. Across four GPUs, PRQ-LF delivers the highest non-blocking throughput in the closed loop and most open-loop conditions, reaching 5.0--20.4$\times$ the blocking BQ at $2^{15}$ threads, so the progress-versus-performance tradeoff assumed by prior GPU queue designs is not fundamental. PRQ-WF leads the high-concurrency draining condition on three GPUs. Queue-backed BFS shows that these gains reduce end-to-end time only when queue processing is a substantial fraction of execution.

\section*{Acknowledgement}
The manuscript was drafted by the authors. ChatGPT 5.6 and Claude Opus 5 were used to identify spelling, and tone issues, with all suggestions reviewed and selectively included to improve clarity and polish.

\balance

\bibliographystyle{ACM-Reference-Format}
\bibliography{acm_queue_references}

\end{document}